\documentclass[12pt]{article}
\usepackage{amssymb,amsthm,amsmath,latexsym}
\usepackage{url}
\usepackage{color}
\usepackage{comment}
\newtheorem{thm}{Theorem}[section]
\newtheorem{prop}[thm]{Proposition}
\newtheorem{lem}[thm]{Lemma}
\newtheorem{cor}[thm]{Corollary}
\theoremstyle{remark}
\newtheorem{rem}[thm]{Remark}
\newcommand{\FF}{\mathbb{F}}
\newcommand{\ZZ}{\mathbb{Z}}
\newcommand{\0}{\mathbf{0}}
\newcommand{\1}{\mathbf{1}}
\newcommand{\ww}{\omega}
\newcommand{\vv}{\omega^2}
\newcommand{\cC}{\mathcal{C}}
\newcommand{\cD}{\mathcal{D}}

\DeclareMathOperator{\rank}{rank}

\begin{document}
\title{
On the classification of quaternary optimal Hermitian LCD codes
}

\author{
Makoto Araya\thanks{Department of Computer Science,
Shizuoka University,
Hamamatsu 432--8011, Japan.
email: {\tt araya@inf.shizuoka.ac.jp}}
and 
Masaaki Harada\thanks{
Research Center for Pure and Applied Mathematics,
Graduate School of Information Sciences,
Tohoku University, Sendai 980--8579, Japan.
email: {\tt mharada@tohoku.ac.jp}.}
}

\maketitle

\begin{abstract}
We propose a method for a classification of
quaternary Hermitian LCD codes having large minimum weights.
As an example, we give a classification of
quaternary optimal Hermitian LCD codes of dimension $3$.
\end{abstract}

\section{Introduction}\label{sec:1}

Linear complementary dual (LCD for short) 
codes are codes that intersect with their dual codes
trivially.
LCD codes were introduced by Massey~\cite{Massey} and 
gave an optimum linear
coding solution for the two user binary adder channel.
Recently, much work has been done concerning LCD codes
for both theoretical and practical reasons
(see e.g.~\cite{CG}, \cite{CMTQP}, \cite{GOS}, \cite{LLGF}
and the references given therein).
For example,
if there is a quaternary Hermitian LCD $[n,k,d]$ code,
then there is a maximal entanglement 
entanglement-assisted quantum error-correcting $[[n,k,d;n-k]]$ code 
(see e.g.~\cite{LLGF}).
From this point of view,
quaternary Hermitian LCD codes
play an important role in the study of
maximal entanglement entanglement-assisted quantum error-correcting codes.
In addition, Carlet, Mesnager, Tang, Qi and Pellikaan~\cite{CMTQP}
showed that 
any code over $\FF_q$ is equivalent to some Euclidean LCD code
for $q \ge 4$ and
any code over $\FF_{q^2}$ is equivalent to some Hermitian LCD code
for $q \ge 3$,
where $\FF_q$ denotes the finite field of order $q$ and $q$ is a prime power.
This is also a motivation of our study of 
quaternary Hermitian LCD codes.

It is a fundamental problem to determine the largest minimum
weight $d_4(n,k)$ among all quaternary Hermitian LCD 
$[n,k]$ codes and classify 
quaternary optimal Hermitian LCD $[n,k,d_4(n,k)]$ codes
for a given pair $(n,k)$.
It was shown that
$d_4(n,2)=\lfloor \frac{4n}{5} \rfloor$
if  $n \equiv 1,2,3 \pmod 5$ and
$d_4(n,2)=\lfloor \frac{4n}{5} \rfloor-1$
otherwise
for $n \ge 3$~\cite{Li} and~\cite{LLGF}.
Recently, it has been shown that
$d_4(n,3)=
\lfloor \frac{16n}{21} \rfloor$ if
$n \equiv 5,9,13,17,18 \pmod{21}$ and
$d_4(n,3)=
\lfloor \frac{16n}{21} \rfloor-1$
otherwise
for $n \ge 6$~\cite{AHS} and~\cite{LLGF}.
More recently,
Ishizuka~\cite{I} has completed a classification of
quaternary optimal Hermitian LCD codes of dimension $2$.

Araya, Harada and Saito~\cite{AHS} gave some conditions on the
nonexistence of certain quaternary Hermitian LCD codes having large minimum
weights (\cite[Theorem~9]{AHS}).
The aim of this note is to propose a method for a classification of
quaternary Hermitian LCD codes having large minimum weights
by following the same line as in the proof of~\cite[Theorem~9]{AHS}.
As an example, we give a classification of
quaternary optimal Hermitian LCD $[n,3,d_4(n,3)]$ codes 
for arbitrary $n$.
We also give an alternative classification of
quaternary optimal Hermitian LCD $[n,2,d_4(n,2)]$ codes 
and a classification of
quaternary near-optimal Hermitian LCD $[n,2,d_4(n,2)-1]$ codes
for arbitrary $n$.

\section{Preliminaries}\label{sec:2}
In this section, 
we give some definitions, notations and basic results used in this
note.

We denote the finite field of order $4$
by $\FF_4=\{ 0,1,\ww , \vv  \}$, where $\omega^2 = \omega +1$.
For any element $\alpha \in \FF_{4}$, the conjugation of $\alpha$ is
defined as $\overline{\alpha}=\alpha^2$.
Throughout this note, we use the following notations.
Let $\0_{s}$ and $\1_{s}$ denote the zero vector and the all-one vector of 
length $s$, respectively.
Let $O$ denote the zero matrix of appropriate size.
Let $I_k$ denote the identity matrix of order $k$.
Let $A^T$ denote the transpose of a matrix $A$.
For a $k \times n$ matrix $A=(a_{ij})$, 
the conjugate matrix of $A$ is defined as
$\overline{A}=(\overline{a_{ij}})$.
For a positive integer $s$ and a $k \times n$ matrix $A$, 
we denote by $A^{(s)}$ 
the $k \times ns$ matrix
$
\left(
\begin{array}{cccccccc}
A & \cdots & A
\end{array}
\right).
$

A {\em quaternary} $[n,k]$ {\em code} $C$
is a $k$-dimensional vector subspace of $\FF_4^n$.
The parameters $n$ and $k$
are called the {\em length} and {\em dimension} of $C$, respectively.
A generator matrix of a quaternary $[n,k]$ code $C$ is a $k \times n$
matrix such that the rows of the matrix generate $C$.
The {\em weight}
of a vector $x \in \FF_4^n$ is
the number of non-zero components of $x$.
A vector of $C$ is called a {\em codeword} of $C$.
The minimum non-zero weight of all codewords in $C$ is called
the {\em minimum weight} of $C$. A quaternary $[n,k,d]$ code
is a quaternary $[n,k]$ code with minimum weight $d$.
Two quaternary $[n,k]$ codes $C$ and $C'$ are
{\em equivalent}, denoted $C \cong C'$,
if there is an $n \times n$ monomial matrix $P$ over $\FF_4$ with
$C' = \{ x P \mid x \in C\}$.
For any quaternary $[n,k,d]$ code, 
the Griesmer bound is given by
$
n \ge \sum_{i=0}^{k-1} \left\lceil \frac{d}{4^i}\right\rceil$.
Throughout this note, we use the following notation:
\begin{equation*}\label{eq:Gb2}
g_4(n,k)=\max\left\{d \in \mathbb{Z}_{\ge 0} ~\middle|~
n \ge \sum_{i=0}^{k-1} \left\lceil \frac{d}{4^i}\right\rceil\right\},
\end{equation*}
where $\ZZ_{\ge 0}$ denotes the set of nonnegative integers.

The {\em Hermitian dual} code $C^\perp$ of
a quaternary $[n,k]$ code $C$ is defined as:
\begin{align*}
C^\perp&=
\{x \in \FF_{4}^n \mid \langle x,y\rangle_H = 0 \text{ for all } y \in C\},
\end{align*}
where
$\langle x,y\rangle_H= \sum_{i=1}^{n} x_i \overline{y_i}$
for $x=(x_1,x_2,\ldots,x_n), y=(y_1,y_2,\ldots,y_n) \in \FF_4^n$.
A quaternary $[n,k]$ code $C$ is called
{\em Hermitian linear complementary dual}
(Hermitian LCD for short)
if $C \cap C^\perp = \{\0_n\}$.
Note that quaternary Hermitian LCD
codes are also called {\em zero radical} codes
(see e.g.~\cite{LLGF}).
Let $d_4(n,k)$ denote the largest minimum weight
among all quaternary Hermitian LCD $[n,k]$ codes.
A quaternary Hermitian LCD $[n,k,d_4(n,k)]$ code is called {\em optimal}.
In this note, we say that a quaternary Hermitian LCD $[n,k,d_4(n,k)-1]$ code
is {\em near-optimal}.
The minimum weight of the Hermitian dual code $C^\perp$ of $C$ is
called the (Hermitian) {\em dual distance} of $C$
and it is denoted by $d^\perp$.

The following characterization gives a criterion for 
quaternary Hermitian LCD codes and is analogous to~\cite[Proposition~1]{Massey}.

\begin{prop}[{\cite[Proposition~3.5]{GOS}}]
Let $C$ be a quaternary code and let $G$ be a generator matrix of $C$.
Then $C$ is a Hermitian LCD code if and only if
$G \overline{G}^T$ is nonsingular.
\end{prop}

Throughout this note, we use the above characterization 
without mentioning this.

A quaternary code $C$ is called
{\em Hermitian self-orthogonal}
if  $C \subset C^\perp$.
A quaternary code $C$ is called {\em even} if the weights
of all codewords of $C$ are even.
A quaternary code $C$ is Hermitian self-orthogonal
if and only if $C$ is even~\cite[Theorem~1]{MOSW}.
In addition, 
a quaternary code $C$ is Hermitian self-orthogonal
if and only if $G \overline{G}^T =O$
for a generator matrix $G$ of $C$.

\section{Background materials}

Let $C$ be a quaternary Hermitian LCD $[n,k,d]$ code.
Define an $[n+1,k,d]$ code $\widehat{C}$ as
$\widehat{C}=\{(x,0) \mid x \in C\}$.
The following lemma was given for binary LCD codes and ternary LCD
codes~\cite[Proposition~3]{AH-C}.
The argument can be applied to quaternary Hermitian LCD codes  trivially.

\begin{lem}[Ishizuka~\cite{I}]
\label{lem:class}
Let $\cC_{n,k,d}$ denote all equivalence classes of quaternary Hermitian LCD $[n,k,d]$
codes.
Let $\cD_{n,k,d}$ denote all equivalence classes of quaternary Hermitian LCD $[n,k,d]$
codes with dual distances $d^\perp\ge 2$.
Let $\widehat{\cC_{n-1,k,d}}$ denote 
all equivalence classes containing $\widehat{C_1},\widehat{C_2},\ldots,
\widehat{C_t}$,
where
$C_1,C_2,\ldots,C_t$ denote representatives of
$\cC_{n-1,k,d}$ and
$t=|{\cC_{n-1,k,d}}|$.
Then 
$\cC_{n,k,d}= \cD_{n,k,d} \cup \widehat{\cC_{n-1,k,d}}$.
 \end{lem}

For a classification of quaternary Hermitian LCD $[n,k,d]$ codes,
by the above lemma, it is sufficient to consider
a classification of quaternary Hermitian LCD $[n,k,d]$ codes
with dual distances $d^\perp \ge 2$.

According to~\cite{LLGF},
we define the $k \times {(\frac{4^k-1}{3})}$ $\mathbb{F}_4$-matrices $S_{k}$
by inductive constructions as follows:
\begin{align*}
S_{1}&=
\begin{pmatrix}
1
\end{pmatrix}, \\
S_{k}&=
\begin{pmatrix}
S_{k-1} & \0_{\frac{4^{k-1}-1}{3}}^T & S_{k-1} & S_{k-1} & S_{k-1}\\
\0_{\frac{4^{k-1}-1}{3}} & 1
& \1_{\frac{4^{k-1}-1}{3}}  & \omega\1_{\frac{4^{k-1}-1}{3}}  
& \omega^2\1_{\frac{4^{k-1}-1}{3}} 
\end{pmatrix} \text{ if } k \ge 2.
\end{align*}
The matrix $S_{k}$ is a generator matrix of the quaternary simplex
$[\frac{4^k-1}{3},k,4^{k-1}]$ code.
It is known that the quaternary simplex $[\frac{4^k-1}{3},k,4^{k-1}]$ 
code is a constant weight code.
More precisely, the code contains codewords of weights
$0$ and $4^{k-1}$ only.
Thus,
for $k \ge 2$, the quaternary simplex 
$[\frac{4^k-1}{3},k,4^{k-1}]$ code is even.
By~\cite[Theorem~1]{MOSW},
the quaternary simplex $[\frac{4^k-1}{3},k,4^{k-1}]$ code is
Hermitian self-orthogonal for $k \ge 2$.

Let $h_{k,i}$ be the $i$-th column of the
$k \times {(\frac{4^k-1}{3})}$ $\mathbb{F}_4$-matrix $S_{k}$.
For a vector $m=(m_1,m_2,\ldots,m_{{\frac{4^k-1}{3}}}) \in \mathbb{Z}_{\ge 0}^{{\frac{4^k-1}{3}}}$, we define a $k \times \sum_{i=1}^{{\frac{4^k-1}{3}}}m_i$ $\mathbb{F}_4$-matrix $G_{k}(m)$, which consists of $m_i$ columns $h_{k,i}$ for each $i$ as follows:
\begin{equation}\label{eq:Gqkm}
G_{k}(m)=
\left(
 h_{k,1} \cdots h_{k,1} h_{k,2} \cdots h_{k,2}
\cdots h_{k,{\frac{4^k-1}{3}}} \cdots h_{k,{\frac{4^k-1}{3}}}
\right).
\end{equation}
Here $m_i=0$ means that no column of $G_{k}(m)$ is $h_{k,i}$.
Throughout this note,
we denote by $C_{k}(m)$ the quaternary code with generator matrix $G_{k}(m)$.

\begin{rem}\label{rem}
 By considering all vectors $m \in \ZZ_{\ge 0}^{{\frac{4^k-1}{3}}}$
with $n=\sum_{i=1}^{{\frac{4^k-1}{3}}}m_i$,
it is possible to find
representatives of all equivalence classes of 
quaternary $[n,k]$ codes with dual distances $d^\perp \ge 2$
as $C_{k}(m)$.
\end{rem}

The following lemma was given for $k =2,3,4$ in~\cite{LLGF}.
The argument can be applied to arbitrary $k$ trivially.

\begin{lem}\label{lem:190228-1}
Suppose that $k \ge 2$ and $s$ is a positive integer.
Let $m=(m_1,m_2,\ldots,m_{{\frac{4^k-1}{3}}})$ be a vector of
 $\mathbb{Z}_{\ge 0}^{{\frac{4^k-1}{3}}}$
 with $n=\sum_{i=1}^{{\frac{4^k-1}{3}}}m_i$.
If $C_{k}(m)$ is a quaternary Hermitian LCD
$[n,k,d]$
code, then the quaternary code $C$ with generator matrix
\[
\left(\begin{array}{cc}
 S_{k}^{(s)} &G_{k}(m)
       \end{array}\right)
\]
 is a quaternary Hermitian LCD $[n+{\frac{4^k-1}{3}} s,k,
 d+4^{k-1}s]$ code.
\end{lem}

The following lemma was given for $k \ge 3$~\cite[Lemma~7]{AHS}.
The argument can be applied to $k=2$ trivially.

\begin{lem}\label{lem:mi-bound}
Suppose that $k \ge 2$.
Let $m=(m_1,m_2,\ldots,m_{{\frac{4^k-1}{3}}})$ be a vector of
$\mathbb{Z}_{\ge 0}^{{\frac{4^k-1}{3}}}$
with $n=\sum_{i=1}^{{\frac{4^k-1}{3}}}m_i$.
If the quaternary LCD $[n,k]$ code $C_{k}(m)$
has minimum weight at least $d$, then
\begin{equation}\label{eq:mi}
 4d-3n \le m_i \le n-\frac{4^{k-1}-1}{3 \cdot 4^{k-2}}d,
\end{equation}
for each $i \in \{1,2,\ldots,{\frac{4^k-1}{3}}\}$.
\end{lem}

The following lemma was given for binary LCD codes and ternary LCD
codes~\cite[Lemmas~4.3 and 4.4]{AHS2}.
By following the same line as in the proof of~\cite[Lemmas~4.3 and 4.4]{AHS2},
we have the following lemma trivially.

\begin{lem}\label{lem:Eq}
Suppose that $\ell \ge 1$ and $k \ge 2$.
Let $C$ and $C'$ be quaternary Hermitian LCD $[n,k]$ codes with
dual distances $d(C^\perp) \ge 2$ and $d({C'}^\perp) \ge 2$.
Suppose that
there are quaternary Hermitian LCD $[n,k]$ codes $D$ and $D'$ satisfying
the following conditions:
\begin{enumerate}
\item $C \cong D$ and $C' \cong D'$,
\item $D$ and $D'$ have generator matrices
\[
G=
\left(\begin{array}{cccc}
 &S_{k}^{(\ell)} &  G_0 &
\end{array}\right) \text{ and }
G'=
\left(\begin{array}{ccccc}
&S_{k}^{(\ell)} &  G'_0 &
\end{array}\right),
\]
where $G_0$ and $G'_0$ are generator matrices of some quaternary Hermitian LCD
$[n-\frac{(4^k-1)\ell}{3},k]$ codes $C_0$ and $C'_0$,
respectively.
\end{enumerate}
Then $C \cong C'$ if and only if $C_0 \cong C'_0$.
\end{lem}

\section{Characterizations of quaternary Hermitian LCD codes}

In the rest of this note, we use the following notation:
\begin{equation}\label{eq:r}
r_{n,k,d}=4^{k-1}n-{\frac{4^k-1}{3}} d, 
\end{equation}
for a given set of parameters $n,k,d$.

\subsection{Theorem~\ref{thm:main} and its proof}

By following the same line as in the proof of Theorem~9 in~\cite{AHS},
we have the following theorem.

\begin{thm}\label{thm:main}
Suppose that $4d-3n \ge 1$ and $4r_{n,k,d} \ge k \ge 2$, where
$r_{n,k,d}$ is the integer defined in~\eqref{eq:r}.
Then there is a one-to-one correspondence between
equivalence classes of quaternary Hermitian LCD $[n,k,d]$ codes
with dual distances $d^\perp \ge 2$ 
and 
equivalence classes of quaternary Hermitian LCD
$[4r_{n,k,d},k,3r_{n,k,d}]$ codes
with dual distances $d^\perp \ge 2$.
\end{thm}
\begin{proof}
Let $C$ be a quaternary $[n,k,d]$ code
with dual distance $d^\perp \ge 2$.
Since $d^\perp \ge 2$, by Remark~\ref{rem}, 
there is a vector $m=(m_1,m_2,\ldots,m_{{\frac{4^k-1}{3}}}) \in
\mathbb{Z}_{\ge 0}^{{\frac{4^k-1}{3}}}$
such that $C \cong C_{k}(m)$ and
$n=\sum_{i=1}^{{\frac{4^k-1}{3}}}m_i$.
Since the minimum weight of $C_{k}(m)$ is $d$, we have
\[4d-3n \le m_i,\]
by Lemma~\ref{lem:mi-bound}.
Thus, the generator matrix $G_{k}(m)$ in~\eqref{eq:Gqkm}
of $C_{k}(m)$
consists of at least $4d-3n$ columns
$h_{k,i}$ for each $i \in\{1,2,\ldots,{\frac{4^k-1}{3}}\}$.
Note that $4d-3n \ge 1$ from the assumption.
Hence, we obtain a matrix $G$ of the following form:
\begin{equation}\label{eq:G}
G=
\begin{pmatrix}
& S_{k}^{(4d-3n)} &  G_0 &
\end{pmatrix},
\end{equation}
by permuting columns of $G_{k}(m)$.
 Here $G_0$ is a $k \times (n- {\frac{(4^k-1)(4d-3n)}{3}})$ matrix,
 noting that
\[
n- {\frac{(4^k-1)(4d-3n)}{3}}
= 4\left(4^{k-1}n-{\frac{4^k-1}{3}}  d\right)=4r_{n,k,d}. 
\]
Since $S_{k}\overline{S_{k}}^T=O$,
we have $G\overline{G}^T=G_0\overline{G_0}^T$.
Since $C$ is Hermitian LCD,  we have
\begin{equation}\label{eq:rank}
4r_{n,k,d}
\ge \rank(G_0) \ge \rank(G_0\overline{G_0}^T)=\rank(G\overline{G}^T)=k.
\end{equation}
Let $C_0$ be the quaternary code with generator matrix $G_0$.
It follows from~\eqref{eq:rank} that 
$C_0$ is a quaternary Hermitian $[4r_{n,k,d},k]$ LCD code.
From the assumption $k \ge 2$,
the quaternary code $C'$ with generator matrix $S_{k}^{(4d-3n)}$ is
a Hermitian self-orthogonal $[n',k,d']$ code, where
\[
n'={\frac{(4^k-1)(4d-3n)}{3}}
\text{ and }
d'=(4d-3n)4^{k-1}.
\]
By Lemma~\ref{lem:190228-1}, 
we have 
\[
d=d_0+d' \text{ and }
d_0=3\left(4^{k-1}n-{\frac{4^k-1}{3}}  d\right)=3r_{n,k,d}.\]
Hence, 
if there is a quaternary Hermitian LCD $[n,k,d]$ code $C$ with
dual distance $d^\perp \ge 2$,
then there is a quaternary Hermitian LCD $[n,k,d]$ code $C'$ such that
$C \cong C'$ and 
$C'$ has generator matrix of form~\eqref{eq:G}.
In addition,
$G_0$ is a generator matrix of some quaternary Hermitian LCD
$[4r_{n,k,d},k,3r_{n,k,d}]$ code.

Now let $C$ and $C'$ be quaternary Hermitian LCD $[n,k,d]$ codes with
dual distances $d(C^\perp) \ge 2$ and $d({C'}^\perp) \ge 2$.
By the above argument, 
there are quaternary Hermitian LCD $[n,k,d]$ codes $D$ and $D'$ satisfying
the following conditions:
\begin{enumerate}
\item $C \cong D$ and $C' \cong D'$,
\item $D$ and $D'$ have generator matrices
\[
G=
\left(\begin{array}{cccc}
&S_{k}^{(4d-3n)} &  G_0 &
\end{array}\right) \text{ and }
G'=
\left(\begin{array}{ccccc}
&S_{k}^{(4d-3n)} &  G'_0 &
\end{array}\right),
\]
where $G_0$ and $G'_0$ are generator matrices of some quaternary Hermitian LCD
$[4r_{n,k,d},k,3r_{n,k,d}]$ codes $C_0$ and $C'_0$,
respectively.
\end{enumerate}
It follows from Lemma~\ref{lem:Eq} that 
$C \cong C'$ if and only if $C_0 \cong C'_0$.
This completes the proof.
\end{proof}

The above theorem
says that for a given set of parameters $n,k,d$ 
a classification of
quaternary Hermitian LCD $[n,k,d]$ codes is obtained from
that of
quaternary Hermitian LCD $[4r_{n,k,d},k,3r_{n,k,d}]$ codes,
where $4r_{n,k,d} \le n$.

\subsection{Modification of Theorem~\ref{thm:main}}

As the next step,
by following the same line as in the proof of~\cite[Theorem~4.7]{AHS2},
we modify Theorem~\ref{thm:main}
to the form which is used easily
by adding some assumption~\eqref{eq:as} on
minimum weights for our study in Section~\ref{sec:4-2}
(Theorem~\ref{thm:main2}).

Assume that we write
\[
n=\frac{4^k-1}{3} s+t,
\]
where $s \in \ZZ_{\ge 0}$
and $t \in \{0,1,\ldots,\frac{4^k-1}{3}-1\}$.
In addition, assume the following:
\begin{equation}\label{eq:as}
\begin{split}
&\text{the minimum weight $d$ is written as}\\
&d(s,t)=4^{k-1}s+\alpha(t),\\
&\text{where $\alpha(t)$ is a constant depending on only $t$.} 
\end{split}
\end{equation}
The condition $4d-3n \ge 1$ in Theorem~\ref{thm:main}
is equivalent to that $s \ge s'_{(\frac{4^k-1}{3} s+t),k,d(s,t)}$, where
\begin{equation}\label{eq:s0}
s'_{(\frac{4^k-1}{3} s+t),k,d(s,t)}= \frac{4r_{(\frac{4^k-1}{3} s+t),k,d(s,t)}-t}{\frac{4^k-1}{3}}+1. 
\end{equation}
From~\eqref{eq:r}, we have
\begin{equation}\label{eq:r2}
\begin{split}
r_{(\frac{4^k-1}{3} s+t),k,d(s,t)}
&=4^{k-1}\left(\frac{4^k-1}{3} s+t\right)-\frac{4^k-1}{3} d(s,t)
\\
&=4^{k-1}t -\frac{4^k-1}{3} \alpha(t).
\end{split}
\end{equation}

From~\eqref{eq:s0} and~\eqref{eq:r2}, we have the following:

\begin{lem}\label{lem:r}
Both $r_{(\frac{4^k-1}{3} s+t),k,d(s,t)}$
and
$s'_{(\frac{4^k-1}{3} s+t),k,d(s,t)}$
depend on only $k,t$ and do not depend on $s$.
\end{lem}

From~\eqref{eq:s0}, we have
\begin{equation}\label{eq:41}
\begin{split}
 4r_{(\frac{4^k-1}{3} s+t),k,d(s,t)}
 &=\frac{4^k-1}{3} \left(s'_{(\frac{4^k-1}{3}
s+t),k,(4^{k-1}s+\alpha(t))}-1 \right)+t.
\end{split}
\end{equation}
From~\eqref{eq:r2} and~\eqref{eq:41}, we have
\begin{equation}\label{eq:42}
\begin{split}
3r_{(\frac{4^k-1}{3} s+t),k,d(s,t)}&=
3\left(4^{k-1}t -\frac{4^k-1}{3} \alpha(t)\right)\\
&=\frac{3}{4}
\left(\frac{4^k-1}{3} \left(s'_{(\frac{4^k-1}{3}
 s+t),k,(4^{k-1}s+\alpha(t))}-1\right)+t
 \right)\\ 
&=4^{k-1}\left(s'_{(\frac{4^k-1}{3} s+t),k,(4^{k-1}s+\alpha(t))}-1
 \right)\\ 
&\qquad +\frac{1}{4}\left(-  {\left(s'_{(\frac{4^k-1}{3}
 s+t),k,(4^{k-1}s+\alpha(t))}-1\right)}+3t \right)\\
&= 4^{k-1} \left(s'_{(\frac{4^k-1}{3}
 s+t),k,(4^{k-1}s+\alpha(t))}-1 \right)+\alpha(t).
\end{split}
\end{equation}
By Lemma~\ref{lem:r}, \eqref{eq:41} and~\eqref{eq:42},
we have the following:

\begin{thm}\label{thm:main2}
Write $n=\frac{4^k-1}{3} s+t \ge k$,
where $s \in \ZZ_{\ge 0}$
and $t \in \{0,1,\ldots,\frac{4^k-1}{3}-1\}$.
Assume that $d$ is written as $d(s,t)=4^{k-1}s+\alpha(t)$,
where $\alpha(t)$ is a constant depending on $t$.
Let $r$ denote 
the integer $r_{(\frac{4^k-1}{3} s+t),k,d(s,t)}$ defined in~\eqref{eq:r}.
Let $s'$ denote
the integer $s'_{(\frac{4^k-1}{3} s+t),k,d(s,t)}$ defined in~\eqref{eq:s0}.
Suppose that $4r \ge k \ge 2$.
Then there is a one-to-one correspondence between
equivalence classes of quaternary Hermitian LCD codes
with dual distances $d^\perp \ge 2$ and parameters
\begin{equation*} 
[4r,k,3r]
=
\left[
 \frac{4^k-1}{3} \left(s'-1 \right)+t,k,
4^{k-1} \left(s'-1 \right)+\alpha(t)
\right]
\end{equation*}
and equivalence classes of quaternary Hermitian LCD code
with dual distances $d^\perp \ge 2$ and parameters
\begin{equation*} 
\left[\frac{4^k-1}{3} s+t,k,4^{k-1}s+\alpha(t)\right],
\end{equation*}
for every integer
$s \ge s'$.
\end{thm}

The above theorem says that for a given set of parameters
$k,t,\alpha(t)$
a classification of
quaternary Hermitian LCD
$[ \frac{4^k-1}{3} \left(s'-1 \right)+t,k,
4^{k-1} \left(s'-1 \right)+\alpha(t)]$
codes 
yields that of quaternary Hermitian LCD
$[\frac{4^k-1}{3} s+t,k,4^{k-1}s+\alpha(t)]$ codes 
for every integer $s \ge s'$.
We remark that the assumption~\eqref{eq:as} on the minimum weight
is automatically satisfied
for our study in Section~\ref{sec:4-2}.

\subsection{Consequence of Theorem~\ref{thm:main2}}

We end this section by giving a consequence of Theorem~\ref{thm:main2}.

\begin{cor}
Write $n=\frac{4^k-1}{3} s+t \ge k$,
where $s \in \ZZ_{\ge 0}$
and $t \in \{0,1,\ldots,\frac{4^k-1}{3}-1\}$.
Assume that $d$ is written as $d(s,t)=4^{k-1}s+\alpha(t)$,
where $\alpha(t)$ is a constant depending on $t$.
Let $r$ denote 
the integer $r_{(\frac{4^k-1}{3} s+t),k,d(s,t)}$ defined in~\eqref{eq:r}.
Let $s'$ denote
the integer $s'_{(\frac{4^k-1}{3} s+t),k,d(s,t)}$ defined in~\eqref{eq:s0}.
Suppose that $4r \ge k \ge 2$.
If there is no quaternary Hermitian LCD code
with dual distance $d^\perp \ge 2$ and parameters
\[
[4r,k,3r]
=
\left[
 \frac{4^k-1}{3} \left(s'-1 \right)+t,k,
4^{k-1} \left(s'-1 \right)+\alpha(t)
\right],
\]
then there is no quaternary Hermitian LCD code
with dual distance $d^\perp \ge 2$ and parameters
\[
\left[\frac{4^k-1}{3} s+t,k,4^{k-1}s+\alpha(t)\right],
\]
for every integer $s$.
\end{cor}
\begin{proof}
For $s \ge s'$, the assertion follows directly from 
Theorem~\ref{thm:main2}.

Suppose that
there is a  quaternary Hermitian LCD
$[\frac{4^k-1}{3} s''+t,k,4^{k-1}s''+\alpha(t)]$
code
with dual distance $d^\perp \ge 2$
for some  $s'' < s'-1$.
Then, by Lemma~\ref{lem:190228-1},
there is a  quaternary Hermitian LCD
$[\frac{4^k-1}{3}(s'-1)+t,k,4^{k-1}(s'-1)+\alpha(t)]$
code with dual distance $d^\perp \ge 2$.
This is a contradiction.
\end{proof}

The above corollary is an improvement of~\cite[Theorem~9 (ii)]{AHS}
by adding some assumption~\eqref{eq:as} on minimum weights.

\section{Quaternary optimal Hermitian LCD codes}
\label{sec:4-2}

In this section, by Theorem~\ref{thm:main2}, we give a classification of
quaternary optimal Hermitian LCD codes of dimension $2$ and
a classification of
quaternary optimal Hermitian LCD codes of dimension $3$.

\subsection{Classification method}
\label{sec:M}

Here we suppose that $k \in \{2,3\}$.
As described in Remark~\ref{rem},
it is possible to find
representatives of all equivalence classes of quaternary Hermitian
LCD $[n,k]$ codes with dual distances $d^\perp \ge 2$
as $C_{k}(m)$, by considering all vectors
$m=(m_1,m_2,\ldots,m_{\frac{4^k-1}{3}}) \in
\mathbb{Z}_{\ge 0}^{\frac{4^k-1}{3}}$ satisfying
$n=\sum_{i=1}^{\frac{4^k-1}{3}}m_i$
and the condition~\eqref{eq:mi}.
In addition, any quaternary $[n,k,d]$ code is equivalent to some
code with generator matrix of form 
$
\left(\begin{array}{cc}
I_k &  A 
\end{array}\right),
$
where $A$ is a $k \times (n-k)$ matrix and the weight of the first row
of $A$ is exactly $d-1$.
Hence, we may assume without loss of generality that
\begin{align}\label{eq:k2}
 m_1 \ge 1, m_2 \ge 1 &\text{ and }\sum_{i \in {\mathcal S}_k}m_i = d
 \text{ if } k=2,\\\label{eq:k3}
m_1 \ge 1, m_2 \ge 1, m_6 \ge 1 &\text{ and } \sum_{i \in {\mathcal S}_k}m_i = d
 \text{ if } k=3,
\end{align}
where ${\mathcal S}_k$ denotes the support of the first row of $S_k$.
In this way, 
we found all quaternary Hermitian 
LCD $[n,k,d]$ codes which must be checked further for equivalences.
For calculations of determinants of $G\overline{G}^T$
for generator matrices $G$,
the NTL function {\tt determinant}~\cite{ntl} was used. 
To test equivalence of quaternary codes,
we used the algorithm given in~\cite[Section 7.3.3]{KO} as follows.
For a quaternary $[n,k]$ code $C$, define
the digraph $\Gamma(C)$ with vertex set $V$ and arc set $A$, where
\begin{align*}
V&=
C \cup (\{1,2,\dots,n\}\times (\FF_4 \setminus \{0\})),
\\
A&=
\{(c,(j,c_j))\mid c=(c_{1},c_2,\ldots,c_{n}) \in C,  j \in \{1,2,\ldots,n\}\}
\\& \qquad
\cup \{((j,y),(j,\omega y))\mid
j \in \{1,2,\ldots,n\},\ y \in \FF_4 \setminus \{0\}\}.
\end{align*} 
Then, two quaternary $[n,k]$ codes $C$ and $C'$ are equivalent
if and only if $\Gamma(C)$ and $\Gamma(C')$  are isomorphic.
We used {\sc nauty}~\cite{nauty}
for digraph isomorphism testing.

All computer calculations in this section were
done by programs in the language C.
In addition, few verification was done by {\sc Magma}~\cite{Magma}.
Let $\cC_{n,k,d}$ denote our equivalence classes of quaternary
Hermitian LCD $[n,k,d]$ codes with dual distances $d^\perp \ge 2$
obtained by the above method.
Especially, 
we verified by {\sc Magma} that 
$C$ is a quaternary
Hermitian LCD $[n,k,d]$ code with dual distance $d^\perp \ge 2$
for $C \in \cC_{n,k,d}$, and
$C$ and $C'$ are inequivalent
for $C,C' \in \cC_{n,k,d}$ with $C \ne C'$.

\subsection{Quaternary optimal Hermitian $[n,2]$ LCD codes}

The largest minimum weights $d_4(n,2)$ were determined in~\cite{AHS},
where $d_4(n,2)$ are listed in Table~\ref{Tab:4-2}.
Recently, Ishizuka~\cite{I} has completed a classification of 
quaternary optimal Hermitian LCD codes of dimension $2$.
Here we present an alternative approach to the classification by using
Theorem~\ref{thm:main2}.

For $n \ge 2$,
write $n=5s+t$, where $s \in \ZZ_{\ge 0}$ and $t \in \{0,1,\ldots,4\}$.
Let $r=r_{5s+t,2,d_4(5s+t,2)}$ and
$s'=s'_{5s+t,2,d_4(5s+t,2)}$ be the integers defined
in~\eqref{eq:r} and~\eqref{eq:s0}, respectively.
For each $5s+t$, we list $d_4(5s+t,2)$, $s'$ and $r$ in Table~\ref{Tab:4-2}.
Then $d_4(5s+t,2)$ is written as 
$4s+\alpha(t)$, where $\alpha(t)$ is a constant depending on only $t$.
Since $d_4(5s+t,2)$ satisfies 
the assumption~\eqref{eq:as} in Theorem~\ref{thm:main2},
we have the following:

\begin{prop}\label{prop:2}
There is a one-to-one correspondence between
equivalence classes of quaternary Hermitian LCD
$[4r,2,3r]$ codes
with dual distances $d^\perp \ge 2$
and equivalence classes of quaternary Hermitian LCD
$[5s+t,2,d_4(5s+t,2)]$ code
with dual distances $d^\perp \ge 2$ for every integer
$s \ge s'$, where $r$ and $s'$ are listed in Table~\ref{Tab:4-2}.
\end{prop}

\begin{table}[thbp]
\caption{$d_4(n,2)$, $s'$ and $r$}
\label{Tab:4-2}
\begin{center}
{\small
\begin{tabular}{c|c|c|c||c|c|c|c}
\noalign{\hrule height0.8pt}
 $n$ & $d_4(n,2)$ & $s'$ & $r$&
 $n$ & $d_4(n,2)$ & $s'$ & $r$  \\
\hline
 $5s$   &$ 4s-1$&$5$ & $5$ & $5s+3$ &$ 4s+2$&$2$ & $2$ \\
 $5s+1$ &$ 4s $&$4$  & $4$ & $5s+4$ &$ 4s+2$&$5$ & $6$ \\
 $5s+2$ &$ 4s+1$&$3$ & $3$ & &&&\\
\noalign{\hrule height0.8pt}
\end{tabular}
}
\end{center}
\end{table}

By the method given in Section~\ref{sec:M},
our computer search completed a classification of 
all quaternary optimal Hermitian LCD $[4r,2,3r]$ codes
with dual distances $d^\perp \ge 2$
for $r$ listed in Table~\ref{Tab:4-2}.
The numbers $N_4(4r,2)$ of all inequivalent
quaternary optimal Hermitian LCD $[4r,2,3r]$ codes
with dual distances $d^\perp \ge 2$ are listed in Table~\ref{Tab:2}.
In addition, our computer search completed a classification of 
all quaternary optimal Hermitian LCD $[5s+t,2,d_4(5s+t,2)]$ codes
with dual distances $d^\perp \ge 2$
for $s <s'$, where $s'$ is given in Table~\ref{Tab:4-2}.
The numbers $N_4(5s+t,2)$ of all inequivalent
quaternary optimal Hermitian LCD $[5s+t,2,d_4(5s+t,2)]$ codes
with dual distances $d^\perp \ge 2$ are also listed in Table~\ref{Tab:2}.
From Proposition~\ref{prop:2} and Table~\ref{Tab:2},
we have the following:

\begin{prop}\label{prop:4-2-2}
\begin{enumerate}
\item Suppose that $t \in \{0,1\}$.
Then there are $2$ inequivalent quaternary optimal Hermitian LCD $[5s+t,2]$
codes with dual distances $d^\perp\ge 2$
for every integer $s \ge 2$.

\item Suppose that $t \in \{2,3\}$.
Then there is a unique quaternary optimal Hermitian LCD $[5s+t,2]$
codes with dual distances $d^\perp\ge 2$, up to equivalence, 
for every integer $s \ge 0$.

\item There are $5$ inequivalent quaternary optimal Hermitian LCD $[5s+4,2]$
codes with dual distances $d^\perp\ge 2$ for every integer $s \ge 3$.
\end{enumerate}
\end{prop}

Of course, the above classification coincides with that given in~\cite{I}.

\begin{table}[thbp]
\caption{$N_4(n,2)$}
\label{Tab:2}
\begin{center}
{\small
\begin{tabular}{c|l|lll}
 \noalign{\hrule height0.8pt}
 $n$ & \multicolumn{1}{c|}{$N_4(4r,2)$}
 & \multicolumn{3}{c}{$N_4(5s+t,2)$ $(s < s')$} \\
\hline
$5s$   & $N_4(20,2)=2$ &$N_4(5,2)=1$& $N_4(10,2)=2$ &$N_4(15,2)=2$\\
$5s+1$ & $N_4(16,2)=2$ &$N_4(6,2)=1$& $N_4(11,2)=2$   \\
$5s+2$ & $N_4(12,2)=1$ &$N_4(2,2)=1$& $N_4(7,2)=1$   \\
$5s+3$ & $N_4(8,2)=1$  &$N_4(3,2)=1$& \\
$5s+4$ & $N_4(24,2)=5$ &$N_4(4,2)=1$& $N_4(9,2)=3$ &$N_4(14,2)=4$\\
	&& $N_4(19,2)=5$\\
\noalign{\hrule height0.8pt}
\end{tabular}
}
\end{center}
\end{table}

\subsection{Quaternary Hermitian LCD $[n,2,d_4(n,2)-1]$ codes}

Similar to the previous subsection, 
here we give a classification of 
quaternary near-optimal Hermitian LCD $[n,2,d_4(n,2)-1]$ codes.
Similar to Proposition~\ref{prop:2}, we have the following:

\begin{prop}\label{prop:2d}
There is a one-to-one correspondence between
equivalence classes of quaternary Hermitian LCD
$[4r,2,3r]$ codes
with dual distances $d^\perp \ge 2$
and equivalence classes of quaternary Hermitian LCD
$[5s+t,2,d_4(5s+t,2)-1]$ code
with dual distances $d^\perp \ge 2$ for every integer
$s \ge s'$, where $r$ and $s'$ are listed in Table~\ref{Tab:4-2d}.
\end{prop}

\begin{table}[thbp]
\caption{$d_4(n,2)-1$, $s'$ and $r$}
\label{Tab:4-2d}
\begin{center}
{\small
\begin{tabular}{c|c|c|c||c|c|c|c}
\noalign{\hrule height0.8pt}
 $n$ & $d_4(n,2)-1$ & $s'$ & $r$&
 $n$ & $d_4(n,2)-1$ & $s'$ & $r$  \\
\hline
 $5s$   &$ 4s-2$&$9$ & $10$ & $5s+3$ &$ 4s  $&$10$ & $12$ \\
 $5s+1$ &$ 4s-1$&$8$ & $ 9$ & $5s+4$ &$ 4s+1$&$ 9$ & $11$ \\
 $5s+2$ &$ 4s  $&$7$ & $ 8$ & &&&\\
\noalign{\hrule height0.8pt}
\end{tabular}
}
\end{center}
\end{table}

Our computer search completed a classification of 
all quaternary near-optimal Hermitian LCD $[4r,2,3r]$ codes
with dual distances $d^\perp \ge 2$
for $r$ listed in Table~\ref{Tab:4-2d} and 
all quaternary near-optimal Hermitian LCD $[5s+t,2,d_4(5s+t,2)-1]$ codes
with dual distances $d^\perp \ge 2$
for $s <s'$, where $s'$ is given in Table~\ref{Tab:4-2d}.
In Table~\ref{Tab:2d},
we list 
the numbers $N'_4(4r,2)$ of the inequivalent quaternary
near-optimal Hermitian LCD $[4r,2,3r]$ codes and
the numbers $N'_4(5s+t,2)$ of the inequivalent quaternary
near-optimal Hermitian LCD $[5s+t,2,d_4(5s+t,2)-1]$ codes.
From Proposition~\ref{prop:2d} and Table~\ref{Tab:2d},
we have the following:

\begin{prop}\label{prop:4-2-2d}
\begin{enumerate}
\item There are $15$ inequivalent quaternary near-optimal Hermitian LCD
       $[5s,2,4s-2]$
codes with dual distances $d^\perp\ge 2$ for every integer $s \ge 7$.
\item There are $7$ inequivalent quaternary near-optimal Hermitian LCD
       $[5s+1,2,4s-1]$
codes with dual distances $d^\perp\ge 2$ for every integer $s \ge 5$.
\item There are $8$ inequivalent quaternary near-optimal Hermitian LCD
       $[5s+2,2,4s]$
codes with dual distances $d^\perp\ge 2$ for every integer $s \ge 5$.
\item There are $22$ inequivalent quaternary near-optimal Hermitian LCD
       $[5s+3,2,4s]$
codes with dual distances $d^\perp\ge 2$ for every integer $s \ge 8$.
\item There are $12$ inequivalent quaternary near-optimal Hermitian LCD
       $[5s+4,2,4s+1]$
codes with dual distances $d^\perp\ge 2$ for every integer $s \ge 6$.
\end{enumerate}
\end{prop}

\begin{table}[thbp]
\caption{$N'_4(n,2)$}
\label{Tab:2d}
\begin{center}
{\small
\begin{tabular}{c|l|lll}
 \noalign{\hrule height0.8pt}
 $n$ & \multicolumn{1}{c|}{$N'_4(4r,2)$}
 & \multicolumn{3}{c}{$N'_4(5s+t,2)$ $(s < s')$} \\
\hline
$5s$   & $N'_4(40,2)=15$ &$N'_4(5,2)=1$ &$N'_4(10,2)=5$ &$N'_4(15,2)=9$ \\
  &&$N'_4(20,2)=11$ &$N'_4(25,2)=13$ &$N'_4(30,2)=14$ \\
  &&$N'_4(35,2)=15$ &&\\
$5s+1$ & $N'_4(36,2)=7$ &$N'_4(6,2)=2$&$N'_4(11,2)=3$&$N'_4(16,2)=5$\\
       &                &$N'_4(21,2)=6$&$N'_4(26,2)=7$&$N'_4(31,2)=7$\\
$5s+2$ & $N'_4(32,2)=8$ &$N'_4(7,2)=2$&$N'_4(12,2)=4$&$N'_4(17,2)=6$ \\
       &                &$N'_4(22,2)=7$&$N'_4(27,2)=8$ \\
$5s+3$ & $N'_4(48,2)=22$&$N'_4(8,2)=2$&$N'_4(13,2)=7$&$N'_4(18,2)=12$ \\
       &                &$N'_4(23,2)=16$&$N'_4(28,2)=18$&$N'_4(33,2)=20$ \\
       &                &$N'_4(38,2)=21$&$N'_4(43,2)=22$& \\
$5s+4$ & $N'_4(44,2)=12$&$N'_4(4,2)=1$&$N'_4(9,2)=2$&$N'_4(14,2)=7$\\
       &                &$N'_4(19,2)=8$&$N'_4(24,2)=10$&$N'_4(29,2)=11$\\
       &                &$N'_4(34,2)=12$&$N'_4(39,2)=12$&\\
\noalign{\hrule height0.8pt}
\end{tabular}
}
\end{center}
\end{table}

\begin{table}[thb]
\caption{$(m_1,m_2,m_3,m_4,m_5)$}
\label{Tab:m}
\begin{center}
{\footnotesize
\begin{tabular}{c|llll}
\noalign{\hrule height0.8pt}
 $(n,d)$ & \multicolumn{4}{c}{$(m_1,m_2,m_3,m_4,m_5)$} \\
\hline
$(32,24)$ &
$(3,8,8,8,5)$ &
$(5,8,5,6,8)$ &
$(6,8,5,7,6)$ &
$(6,8,7,4,7)$ \\&
$(6,8,8,7,3)$ &
$(7,8,4,5,8)$ &
$(8,8,2,7,7)$ &
$(8,8,8,7,1)$ 
\\
\hline
 $(36,27)$ &
$(6,9,7,6,8)$ &
$(6,9,9,6,6)$ &
$(8,9,3,8,8)$ &
$(8,9,8,2,9)$ \\&
$(8,9,8,4,7)$ &
$(8,9,8,5,6)$ &
$(8,9,9,6,4)$ 
\\
\hline
 $(40,30)$ &
$(4,10,10,7,9)$&
$(6,10,7,9,8) $&
$(6,10,7,10,7)$&
$(6,10,9,10,5)$\\&
$(7,10,5,8,10)$&
$(8,10,7,8,7) $&
$(8,10,9,5,8) $&
$(8,10,9,6,7) $\\&
$(9,10,9,6,6) $&
$(9,10,9,8,4) $&
$(9,10,10,2,9)$&
$(10,10,1,9,10)$\\&
$(10,10,9,3,8) $&
$(10,10,10,5,5)$&
$(10,10,10,7,3)$
\\
\hline
 $(44,33)$ &
$(4,11,9,10,10)$&
$(7,11,10,8,8) $&
$(8,11,8,8,9)  $&
$(8,11,11,6,8) $\\&
$(9,11,6,10,8) $&
$(10,11,10,3,10)$&
$(10,11,10,6,7) $&
$(10,11,8,5,10) $\\&
$(10,11,9,8,6)  $&
$(10,11,11,10,2)$&
$(11,11,6,10,6) $&
$(11,11,8,10,4) $
\\
\hline
 $(48,36)$ &
$(5,12,12,11,8)$&
$(8,12,9,8,11) $&
$(8,12,10,11,7)$&
$(9,12,9,12,6) $\\&
$(9,12,10,9,8) $&
$(9,12,10,10,7)$&
$(9,12,11,6,10)$&
$(10,12,6,9,11)$\\&
$(10,12,11,5,10)$&
$(11,12,1,12,12)$&
$(11,12,2,11,12)$&
$(11,12,6,8,11) $\\&
$(11,12,7,10,8) $&
$(11,12,10,12,3)$&
$(11,12,11,4,10)$&
$(12,12,3,9,12) $\\&
$(12,12,4,9,11) $&
$(12,12,5,7,12) $&
$(12,12,6,7,11) $&
$(12,12,7,9,8)  $\\&
$(12,12,9,10,5) $&
$(12,12,10,7,7) $\\
\noalign{\hrule height0.8pt}
\end{tabular}
}
\end{center}
\end{table}

For each $C$ of the quaternary codes listed in Table~\ref{Tab:2d}, 
there is a vector $(m_1,m_2,\ldots,,m_5) \in \ZZ_{\ge 0}^{5}$ with
$C_{2}(m) \cong C$.
For
\[
(n,d)=
(32, 24),
(36, 27),
(40, 30),
(44, 33),
(48, 36),
\]
the corresponding
vectors $(m_1,m_2,\ldots,m_5)$ are listed in Table~\ref{Tab:m}.
For the codes listed in Table~\ref{Tab:2d}, 
let $V_{5s+t}$ be the set of the corresponding vectors
$(m_1,m_2,\ldots,m_5)$
for the inequivalent quaternary
near-optimal Hermitian LCD $[5s+t,2,d_4(5s+t,2)-1]$ codes
with dual distances $d^\perp\ge 2$.
We verified that
$V_{5(s-1)+t}$ is obtained as:
\[
\{(m_1-1,m_2-1,\ldots,m_5-1)\in\ZZ_{\ge 0}^5 \mid
(m_1,m_2,\ldots,m_5) \in V_{5s+t} \}
\]
for $s \le s'$.
The corresponding
vectors $(m_1,m_2,\ldots,,m_5)$ can be also obtained electronically from 
\url{http://www.math.is.tohoku.ac.jp/~mharada/qLCD/}.

\subsection{Quaternary optimal Hermitian $[n,3]$ LCD codes}

The largest minimum weights $d_4(n,3)$ were determined in~\cite{AHS},
where $d_4(n,3)$ are listed in Table~\ref{Tab:4-3}.
In this subsection, we complete
a classification of quaternary optimal Hermitian LCD codes of dimension $3$.

We apply Theorem~\ref{thm:main2} to $k=3$.
For $n \ge 3$,
write $n=21s+t$, where $s \in \ZZ_{\ge 0}$ and $t \in \{0,1,\ldots,20\}$.
Let $r=r_{21s+t,3,d_4(21s+t,3)}$ and
$s'=s'_{21s+t,3,d_4(21s+t,3)}$ be the integers defined
in~\eqref{eq:r} and~\eqref{eq:s0}, respectively.
For each $21s+t$, we list $d_4(21s+t,3)$, $s'$ and $r$ in Table~\ref{Tab:4-3}.
Then $d_4(21s+t,3)$ is written as 
$16s+\alpha(t)$, where $\alpha(t)$ is a constant depending on only $t$.
Since $d_4(21s+t,3)$  satisfies 
the assumption~\eqref{eq:as} in Theorem~\ref{thm:main2},
we have the following:

\begin{prop}\label{prop:3}
There is a one-to-one correspondence between
equivalence classes of quaternary Hermitian LCD
$[4r,3,3r]$ codes
with dual distances $d^\perp \ge 2$
 and equivalence classes of quaternary Hermitian LCD
 $[21s+t,3,d_4(21s+t,3)]$ code
with dual distances $d^\perp \ge 2$ for every integer
$s \ge s'$.
\end{prop}

\begin{table}[thbp]
\caption{$d_4(n,3)$, $s'$ and $r$}
\label{Tab:4-3}
\begin{center}
{\small
\begin{tabular}{c|c|c|c||c|c|c|c}
\noalign{\hrule height0.8pt}
 $n$ & $d_4(n,3)$&$s'$  & $r$ &
 $n$ & $d_4(n,3)$&$s'$  & $r$ \\
\hline
 $21s$   & $ 16s-1 $&$5$  & 21 & $21s+11$& $ 16s+7 $&$6$  & 29 \\
 $21s+1$ & $ 16s-1 $&$8$  & 37 & $21s+12$& $ 16s+8 $&$5$  & 24 \\
 $21s+2$ & $ 16s   $&$7$  & 32 & $21s+13$& $ 16s+9 $&$4$  & 19 \\
 $21s+3$ & $ 16s+1 $&$6$  & 27 & $21s+14$& $ 16s+9 $&$7$  & 35 \\
 $21s+4$ & $ 16s+2 $&$5$  & 22 & $21s+15$ &$ 16s+10$&$6$  & 30 \\
 $21s+5$ & $ 16s+3 $&$4$  & 17 & $21s+16$ &$ 16s+11$&$5$  & 25 \\
 $21s+6$ & $ 16s+3 $&$7$  & 33 & $21s+17$ &$ 16s+12$&$4$  & 20 \\
 $21s+7$ & $ 16s+4 $&$6$  & 28 & $21s+18$ &$ 16s+13$&$3$  & 15 \\
 $21s+8$ & $ 16s+5 $&$5$  & 23 & $21s+19$ &$ 16s+13$&$6$  & 31 \\
 $21s+9$ & $ 16s+6 $&$4$  & 18 & $21s+20$& $ 16s+14$&$5$  & 26 \\
 $21s+10$& $ 16s+6 $&$7$  & 34 & &&&\\
 \noalign{\hrule height0.8pt}
\end{tabular}
}
\end{center}
\end{table}

\begin{table}[thbp]
\caption{$N_4(n,3)$}
\label{Tab:3}
\begin{center}
{\footnotesize
\begin{tabular}{c|l|lll}
 \noalign{\hrule height0.8pt}
 $n$ & \multicolumn{1}{c|}{$N_4(4r,3)$}
 & \multicolumn{3}{c}{$N_4(21s+t,3)$ $(s < s')$} \\
\hline
$21s$    & $N_4(84,3)=7$ & $N_4(21,3)=5$ &$N_4(42,3)=7$  &$N_4(63,3)=7$\\
\hline
$21s+1$&$N_4(148,3)=12808$&$N_4(22,3)=1871$&$N_4(43,3)=9793$&$N_4(64,3)=12405$\\
&&$N_4(85,3)=12781$&$N_4(106,3)=12808$&$N_4(127,3)=12808$\\
\hline
$21s+2$&$N_4(128,3)=318$&$N_4(23,3)=135$&$N_4(44,3)=288$&$N_4(65,3)=318$\\
&&$N_4(86,3)=318$&$N_4(107,3)=318$\\
\hline
$21s+3$&$N_4(108,3)=147$&$N_4(24,3)=73$&$N_4(45,3)=138$&$N_4(66,3)=147$\\
&&$N_4(87,3)=147$\\
\hline
$21s+4$  & $N_4(88,3)=4$ & $N_4(4,3)=0$  & $N_4(25,3)=4$ &$N_4(46,3)=4$ \\
         &               &$N_4(67,3)=4$&&\\
\hline
$21s+5$  & $N_4(68,3)=1$ & $N_4(5,3)=0$  &$N_4(26,3)=1$ &$N_4(47,3)=1$\\
\hline
$21s+6$&$N_4(132,3)=2162$&$N_4(6,3)=2$&$N_4(27,3)=937$&$N_4(48,3)=1948$\\
&&$N_4(69,3)=2145$&$N_4(90,3)=2162$&$N_4(111,3)=2162$\\
\hline
$21s+7$&$N_4(112,3)=44$&$N_4(7,3)=1$&$N_4(28,3)=30$&$N_4(49,3)=44$\\
&&$N_4(70,3)=44$&$N_4(91,3)=44$\\
\hline
$21s+8$  & $N_4(92,3)=23$& $N_4(8,3)=1$  &$N_4(29,3)=18$ &$N_4(50,3)=23$\\
         &               &$N_4(71,3)=23$&&\\
\hline
$21s+9$  & $N_4(72,3)=1$ & $N_4(9,3)=1$  &$N_4(30,3)=1$ &$N_4(51,3)=1$\\
\hline
$21s+10$&$N_4(136,3)=947$&$N_4(10,3)=13$&$N_4(31,3)=589$&$N_4(52,3)=889$\\
&&$N_4(73,4)=947$&$N_4(94,3)=947$&$N_4(115,3)=947$\\
\hline
 $21s+11$&$N_4(116,3)=318$&$N_4(11,3)=13$&$N_4(32,3)=220$&$N_4(53,4)=309$\\
 &&$N_4(74,3)=318$&$N_4(95,3)=318$\\
\hline
$21s+12$ & $N_4(96,3)=7$  & $N_4(12,3)=2$& $N_4(33,3)=7$	 &$N_4(54,3)=7$\\
         &               &$N_4(75,3)=7$&&\\
\hline
$21s+13$ & $N_4(76,3)=4$ & $N_4(13,3)=2$ &$N_4(34,3)=4$ &$N_4(55,3)=4$\\
\hline
$21s+14$&$N_4(140,3)=5736$&$N_4(14,3)=156$&$N_4(35,3)=3562$&$N_4(56,3)=5398$\\
&&$N_4(77,3)=5709$&$N_4(98,3)=5736$&$N_4(119,3)=5736$\\
\hline
$21s+15$&$N_4(120,3)=147$&$N_4(15,3)=28$&$N_4(36,3)=118$&$N_4(57,3)=147$\\
&&$N_4(78,3)=147$&$N_4(99,3)=147$\\
\hline
$21s+16$ & $N_4(100,3)=44$ & $N_4(16,3)=10$& $N_4(37,3)=39$	 &$N_4(58,3)=44$\\
         &               &$N_4(79,3)=44$&&\\
\hline
$21s+17$ & $N_4(80,3)=1$ & $N_4(17,3)=1$ &$N_4(38,3)=1$ &$N_4(59,3)=1$\\
\hline
$21s+18$ & $N_4(60,3)=1$ & $N_4(18,3)=1$ &$N_4(39,3)=1$\\
\hline
$21s+19$&$N_4(124,3)=947$&$N_4(19,3)=196$&$N_4(40,3)=774$&$N_4(61,3)=930$\\
&&$N_4(82,3)=947$&$N_4(103,3)=947$\\
\hline
$21s+20$&$N_4(104,3)=23$&$N_4(20,3)=10$&$N_4(41,3)=23$&$N_4(62,3)=23$\\
&&$N_4(83,3)=23$ \\
 \noalign{\hrule height0.8pt}
\end{tabular}
}
\end{center}
\end{table}

By the method given in Section~\ref{sec:M},
our computer search completed a classification of 
all quaternary optimal Hermitian LCD $[4r,3,3r]$ codes
with dual distances $d^\perp \ge 2$
for $r$ listed in Table~\ref{Tab:4-3}.
The numbers $N_4(4r,3)$ of all inequivalent
quaternary optimal Hermitian LCD $[4r,3,3r]$ codes
with dual distances $d^\perp \ge 2$ are listed in Table~\ref{Tab:3}.
In addition,
our computer search completed a classification of 
all quaternary optimal Hermitian LCD $[21s+t,3,d_4(21s+t,3)]$ codes
with dual distances $d^\perp \ge 2$
for $s <s'$, where $s'$ is given in Table~\ref{Tab:4-3}.
The numbers $N_4(21s+t,3)$ of all inequivalent
quaternary optimal Hermitian LCD $[21s+t,3,d_4(21s+t,3)]$ codes
with dual distances $d^\perp \ge 2$ are also listed in Table~\ref{Tab:3}.
From Proposition~\ref{prop:3} and Table~\ref{Tab:3},
we have the following:

\begin{thm}\label{thm:4-3-2}
\begin{enumerate}
\item 
Suppose that $t \in\{0,12\}$.
Then there are $7$ inequivalent quaternary optimal Hermitian LCD $[21s+t,3]$
codes with dual distances $d^\perp\ge 2$
for every integer $s \ge 2$ if $t=0$ and  $s \ge 1$ if $t=12$.

\item 
There are $12808$ inequivalent quaternary optimal Hermitian LCD $[21s+1,3]$
codes with dual distances $d^\perp\ge 2$
for every integer $s \ge 5$.
\item 
Suppose that $t \in\{2,11\}$.
Then there are $318$ inequivalent quaternary optimal Hermitian LCD $[21s+t,3]$
codes with dual distances $d^\perp\ge 2$
for every integer $s \ge 3$.
\item 
Suppose that $t \in \{3,15\}$.
Then there are $147$ inequivalent quaternary optimal Hermitian LCD $[21s+t,3]$
codes with dual distances $d^\perp\ge 2$
for every integer $s \ge 3$ if $t=3$ and  $s \ge 2$ if $t=15$.

\item Suppose that $t \in\{4,13\}$.
Then there are $4$ inequivalent quaternary optimal Hermitian LCD $[21s+t,3]$
codes with dual distances $d^\perp\ge 2$
for every integer $s \ge 1$.

\item Suppose that $t \in\{5,9,17,18\}$.
Then there is a unique quaternary optimal Hermitian LCD $[21s+t,3]$
code with dual distance $d^\perp\ge 2$, up equivalence,
for every integer
$s \ge 1$ if $t=5$ and $s \ge 0$ if $t \in\{9,17,18\}$.

\item 
There are $2162$ inequivalent quaternary optimal Hermitian LCD $[21s+6,3]$
codes with dual distances $d^\perp\ge 2$
for every integer $s \ge 4$.

\item 
Suppose that $t \in\{7,16\}$.
Then there are $44$ inequivalent quaternary optimal Hermitian LCD $[21s+t,3]$
codes with dual distances $d^\perp\ge 2$
for every integer $s \ge 2$.
\item 
Suppose that $t \in\{8,20\}$.
Then there are $23$ inequivalent quaternary optimal Hermitian LCD $[21s+t,3]$
codes with dual distances $d^\perp\ge 2$
for every integer $s \ge 2$ if $t=8$ and  $s \ge 1$ if $t=20$.
      
\item 
Suppose that $t \in\{10, 19\}$.
Then there are $947$ inequivalent quaternary optimal Hermitian LCD $[21s+t,3]$
codes with dual distances $d^\perp\ge 2$
for every integer $s \ge 3$.
\item 
There are $5736$ inequivalent quaternary optimal Hermitian LCD $[21s+14,3]$
codes with dual distances $d^\perp\ge 2$
for every integer $s \ge 4$.

\end{enumerate}
\end{thm}

For each $C$ of the quaternary codes listed in Table~\ref{Tab:3}, 
there is a vector $m \in \ZZ_{\ge 0}^{21}$ with $C_{3}(m) \cong C$.
The vectors $m$ can be obtained electronically from 
\url{http://www.math.is.tohoku.ac.jp/~mharada/qLCD/}.
For the codes listed in Table~\ref{Tab:3}, 
let $V_{21s+t}$ be the set of the corresponding vectors
$(m_1,m_2,\ldots,m_{21})$ for
the inequivalent quaternary optimal Hermitian LCD $[21s+t,3]$
codes with dual distances $d^\perp\ge 2$.
We verified that
$V_{21(s-1)+t}$ is obtained as:
\[
\{(m_1-1,m_2-1,\ldots,m_{21}-1)\in\ZZ_{\ge 0}^{21} \mid
(m_1,m_2,\ldots,m_{21}) \in V_{21s+t} \}
\]
for $s  \le s'$.

\bigskip
\noindent
{\bf Acknowledgments.}
The authors would like to thank Ken Saito for useful discussions.
This work was supported by JSPS KAKENHI Grant Number 19H01802.


\begin{thebibliography}{99}
\bibitem{AH-C}M. Araya and M. Harada, 
On the classification of linear complementary dual codes,
{\sl Discrete Math.}
{\bf 342} (2019), 270--278.

	
\bibitem{AHS} M. Araya, M. Harada and K. Saito,
Quaternary Hermitian linear complementary dual codes,
{\sl IEEE Trans.\ Inform.\ Theory}
{\bf 66} (2020), 2751--2759.

\bibitem{AHS2} M. Araya, M. Harada and K. Saito,	
Characterization and classification of optimal LCD codes,
{\bf 89} (2021), 617--640.
	
\bibitem{Magma}W. Bosma, J. Cannon and C. Playoust,
The Magma algebra system I: The user language,
{\sl J. Symbolic Comput.}
{\bf 24} (1997), 235--265.

\bibitem{CG}
C. Carlet and S. Guilley,
Complementary dual codes for counter-measures to side-channel attacks,
{\sl Adv.\ Math.\ Commun.}
{\bf 10}  (2016),  131--150.

	
\bibitem{CMTQP}
C. Carlet, S. Mesnager, C. Tang, Y. Qi and R. Pellikaan,
Linear codes over $\FF_q$ are equivalent to Hermitian LCD codes for $q >3$,
{\sl IEEE\ Trans.\ Inform.\ Theory}
{\bf 64}  (2018),  3010--3017.




\bibitem{GOS} C. G\"uneri, B. \"Ozkaya and P. Sol\'e, 
Quasi-cyclic complementary dual codes,
{\sl Finite Fields Appl.}
{\bf 42} (2016), 67--80.

\bibitem{I} K. Ishizuka,
Classification of optimal quaternary Hermitian LCD codes of
dimension $2$,
{\sl J. Algebra Comb.\ Discrete Struct.\ Appl.}
{\bf 7} (2020), 229--236.

\bibitem{KO}P. Kaski and  P.R.J. \"Osterg\aa rd,
{\sl Classification Algorithms for Codes and Designs},
Springer, Berlin, 2006.

	

\bibitem{Li}R. Li, 
Research on quantum codes and self-orthogonal codes,
Postdoctor work report, Xi'an Jiaotong University (2008).
	
\bibitem{LLGF}L. Lu, R. Li, L. Guo and Q. Fu, 
Maximal entanglement entanglement-assisted quantum codes 
constructed from linear codes,
{\sl Quantum Inf.\ Process.}
{\bf 14} (2015), 165--182.

\bibitem{MOSW}
F.J. MacWilliams, A.M. Odlyzko, N.J.A. Sloane and H.N. Ward,
Self-dual codes over GF(4),
{\sl J. Combin.\ Theory Ser.~A}
{\bf 25} (1978), 288--318.

\bibitem{Massey}J.L. Massey, 
Linear codes with complementary duals,
{\sl Discrete Math.}
{\bf 106/107} (1992), 337--342.

\bibitem{nauty} B.D. McKay and A. Piperno,
Practical graph isomorphism, II,
{\sl J. Symbolic Comput.}
{\bf 60} (2014), 94--112.
	
\bibitem{ntl} V. Shoup,
NTL: A Library for doing Number Theory, 
Available online at \url{http://www.shoup.net/ntl/}.

	
\end{thebibliography}
\end{document}